\documentclass[journal]{IEEEtran}
\usepackage{listings}

\usepackage{amsmath,amssymb,amsthm}
\usepackage{graphicx}
\usepackage{booktabs}
\usepackage{algorithm}
\usepackage{algpseudocode}
\usepackage{url}
\usepackage{cite}
\usepackage{xcolor}
\usepackage{tikz}
\usetikzlibrary{positioning,arrows.meta,shapes.geometric}
\usetikzlibrary{matrix}
\usetikzlibrary{calc}
\usetikzlibrary{positioning,shapes.geometric,arrows.meta}
\newtheorem{definition}{Definition}
\newtheorem{theorem}{Theorem}

\title{MTTR-A: Measuring Cognitive Recovery Latency in Multi-Agent Systems}
\author{
Barak Or,~\IEEEmembership{Member, IEEE}
\IEEEcompsocitemizethanks{
\IEEEcompsocthanksitem Barak Or is with the Office of the CEO, MetaOr Artificial Intelligence,
Haifa 3349602, Israel, and also with the Google--Reichman Tech School,
Reichman University, Herzliya 4610101, Israel.
\protect\\
E-mail: barakorr@gmail.com
}
}

\begin{document}
\maketitle

% ================= Abstract =================
\begin{abstract}
Reliability in multi-agent systems (MAS) built on large language models is increasingly limited by cognitive failures rather than infrastructure faults. Existing runtime monitoring tools describe failures but do not quantify how quickly distributed reasoning recovers once coherence is lost \cite{bhaskhar2023explainable}. We introduce MTTR-A (Mean Time-to-Recovery for Agentic Systems), a runtime reliability metric that measures cognitive recovery latency in MAS. MTTR-A adapts classical dependability theory to agentic orchestration, capturing the time required to detect reasoning drift and restore coherent operation. We further define complementary metrics, including MTBF and a normalized recovery ratio (NRR), and establish theoretical bounds linking recovery latency to long-run cognitive uptime. Using a LangGraph-based benchmark with simulated drift and reflex recovery, we empirically demonstrate measurable recovery behavior across multiple reflex strategies. This work establishes a quantitative foundation for runtime cognitive dependability in distributed agentic systems. More broadly, MTTR-A provides a system-level reliability metric applicable to distributed AI systems operating under partial observability, decentralized decision-making, and non-deterministic reasoning dynamics.
\end{abstract}

\begin{IEEEkeywords}
Multi-agent systems, cognitive reliability, agentic orchestration, runtime stability, fault tolerance.
\end{IEEEkeywords}

\IEEEpeerreviewmaketitle

\section{Introduction}
Classical system engineering defines reliability through metrics such as availability, mean time between failures (MTBF), and mean time to recovery (MTTR) \cite{bass2020sre}.
These indicators were developed for deterministic infrastructures-servers, networks, or control loops-where recovery denotes restoring a failed component to service. 
In contrast, failures in modern \emph{agentic} AI systems are often cognitive: agents may continue running while reasoning drifts, memories desynchronize, or coordination deteriorates.

As multi-agent systems (MAS) built on large language models (LLMs) become increasingly autonomous and interconnected, such cognitive faults cannot be captured by conventional runtime monitoring or uptime metrics. Recent surveys highlight similar reliability and coordination challenges in LLM-based MAS, underscoring the need for runtime control and recovery mechanisms \cite{han2025llmmaschallenges,guo2024llmmasurvey}.

From an AI systems perspective, this raises a foundational question: how should reliability be defined and measured for distributed cognitive systems whose failures are semantic and coordination-driven rather than infrastructural \cite{rawal2021recent}?

A system may remain functionally “up” while its reasoning has diverged or its safety policy degraded. 
This exposes a fundamental gap: how can we measure \emph{recovery} not of infrastructure, but of cognition itself?

To our knowledge, no prior multi-agent AI system defines a quantitative runtime metric for
cognitive recovery; existing approaches focus on runtime monitoring or schema enforcement, not on
recovery latency itself \cite{sun2024openagents}.
Industry reports concentrate on model capability and alignment rather than orchestration-level recovery, 
highlighting the absence of runtime reliability metrics in production AI systems 
\cite{bai2022constitutional}.

\paragraph{Motivation for System Level Metricization:}
In realistic deployments, cognitive or coordination failures are inevitable. 
Rather than analyzing each agent in isolation, reliability must be evaluated at the level of the entire orchestration -the MAS as an interacting intelligence network. 
We therefore adapt the classical MTTR metric into an orchestration-level cognitive context, defining \textbf{MTTR-A (Mean Time-to-Recovery for Agentic Systems)} as a measure of how long a distributed reasoning system takes to detect drift and restore coherence.
This perspective treats the multi-agent architecture as a single cognitive organism whose communication and recovery pathways determine runtime resilience \cite{varga2020resilient}. 
The practical objective, as in classical dependability engineering, is to \emph{minimize} MTTR-A through improved synchronization, reflex coordination, and runtime control policies.

Real-world MAS increasingly operate in environments where recovery latency determines safety, performance, and stability. 
In autonomous vehicle fleets, the agents are self-driving vehicles that communicate through vehicle-to-vehicle and infrastructure networks; even a two-second delay in recovering from a coordination or perception fault can trigger cascading collisions or traffic paralysis. 
In aerial drone swarms, the agents are autonomous UAVs performing collaborative missions such as mapping or rescue; prolonged recovery from signal loss or reasoning drift can fragment formation and terminate the mission. 
In industrial robotic cells, the agents are networked manipulators and sensors executing synchronized assembly tasks; delayed resynchronization after a fault can halt production or cause mechanical conflict. 
In distributed financial trading systems, the agents are autonomous trading algorithms interacting through shared market data; slow recovery from an erroneous strategy or drifted policy can amplify instability and cause systemic loss. 
In cybersecurity defense networks, the agents are detection and classification models monitoring distributed endpoints; long recovery from false positives or misclassifications can block legitimate activity and reduce coverage. 
Finally, in modern LLM-based multi-agent AI systems the agents are reasoning modules coordinating tool use, memory, and planning; high cognitive recovery time leads to inconsistent outputs and degraded trust. 
Across these domains, the reliability of the MAS depends on preventing every fault and on how rapidly the system detects, isolates, and restores coherent operation once a fault occurs. 

While our empirical evaluation focuses on LLM-based agents, the proposed reliability framework applies more broadly to multi-agent AI systems, including planning agents, autonomous systems, and hybrid learning-based controllers.

\vspace{0.5em}
In practical operational deployments, such orchestration-level dependability requires not only measurement but also structured runtime control. When agent workflows drift, loop, or violate safety rules, the system must detect, interrupt, and safely recover-leaving an auditable trace of reasoning and actions. This motivates the need for explicit runtime reliability mechanisms in distributed AI systems. Such layer may provide: Open Cognitive Telemetry (OCT), Causal Drift Graph, Trust, Drift, and Recovery Metrics, Policy Hooks and Enforcement, and Audit $\&$ Compliance Interface.

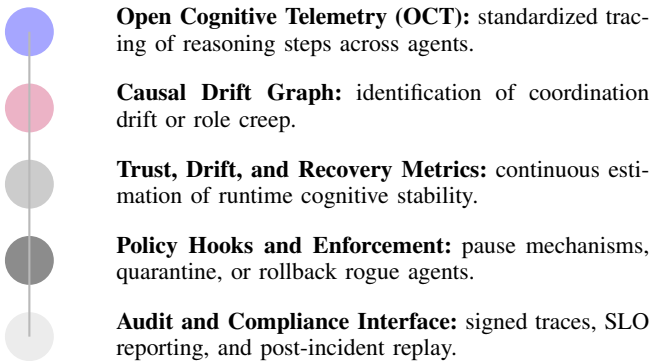
\begin{figure}[ht]
\raggedright
\begin{tikzpicture}[>=Latex, node distance=0.4cm, every node/.style={font=\small, align=left}] % compact spacing

% ---- Styles ----
\tikzstyle{circleblock}=[
    circle, draw=none,
    minimum size=0.65cm,
    inner sep=0pt
]

% ---- Start coordinate ----
\coordinate (xstart) at (0,0);

% ---- Circles (compact vertical layout) ----
\node[circleblock, fill=blue!35, below=0.2cm of xstart, anchor=center] (OCT) {};
\node[circleblock, fill=purple!30, below=0.35cm of OCT] (CDG) {};
\node[circleblock, fill=gray!40, below=0.35cm of CDG] (TDR) {};
\node[circleblock, fill=black!45, below=0.35cm of TDR] (PHE) {};
\node[circleblock, fill=gray!15, below=0.35cm of PHE] (AUDIT) {};

% ---- Vertical connecting line ----
\draw[line width=0.8pt, gray!55] (OCT.center) -- (AUDIT.center);

% ---- Text nodes ----
\node[anchor=west, right=0.7cm of OCT]
(txt1){\parbox[t]{0.8\linewidth}{
\textbf{Open Cognitive Telemetry (OCT):}
standardized tracing of reasoning steps across agents.}};

\node[anchor=west, right=0.7cm of CDG]
(txt2){\parbox[t]{0.8\linewidth}{
\textbf{Causal Drift Graph:}
identification of coordination drift or role creep.}};

\node[anchor=west, right=0.7cm of TDR]
(txt3){\parbox[t]{0.8\linewidth}{
\textbf{Trust, Drift, and Recovery Metrics:}
continuous estimation of runtime cognitive stability.}};

\node[anchor=west, right=0.7cm of PHE]
(txt4){\parbox[t]{0.8\linewidth}{
\textbf{Policy Hooks and Enforcement:}
pause mechanisms, quarantine, or rollback rogue agents.}};

\node[anchor=west, right=0.7cm of AUDIT]
(txt5){\parbox[t]{0.8\linewidth}{
\textbf{Audit and Compliance Interface:}
signed traces, SLO reporting, and post-incident replay.}};

\end{tikzpicture}

\caption{Key operational modules of the Cognitive Reliability Layer, providing runtime monitoring, drift control, and runtime control for multi-agent orchestration.}
\label{fig:cognitive-layer-nobox}
\end{figure}

These operational elements illustrate the broader system context in which MTTR-A becomes essential, as it provides a measure for cognitive recovery and runtime reliability in large-scale MAS. This is a measurable basis for runtime resilience in distributed reasoning systems. 
Importantly, MTTR-A evaluates recovery behavior at the orchestration layer without requiring visibility into the internal reasoning of each agent or the semantics of their communications. It assumes that cognitive faults are inevitable in open, distributed systems, and focuses instead on how quickly the orchestration loop detects, isolates, and restores coherence-treating recovery latency itself as the measurable indicator of cognitive reliability.

This work situates reliability engineering within the domain of agentic cognition, linking classical fault-tolerance principles with reflexive control architectures for LLM-based MAS. 
Combined with a taxonomy of recovery reflexes and an empirical LangGraph benchmark, establishes a foundation for measuring and improving cognitive reliability in MAS.
The empirical benchmark employs a real-text retrieval environment based on the AG News dataset to ensure that reasoning-drift and recovery cycles are observed under natural language conditions.

\subsection{Our Approach: A Reliability Metricization Framework}
This paper positions MTTR-A as a reliability metric for MAS. 
Rather than evaluating semantic accuracy or reasoning quality, we focus on \emph{runtime dependability} -how quickly and predictably a distributed agent workflow detects, isolates, and recovers from cognitive faults. 
It unifies a reflex taxonomy, a latency-based metric family (MTTR-A, MTBF, NRR), and a LangGraph-based telemetry pipeline into a reproducible methodology for benchmarking agentic reliability.
This hierarchical reliability structure treats each agent as a subsystem composed of reflexive recovery components (rollback, replan, retry, approve), 
while the orchestration layer governs overall system-level dependability.

\subsection{Contributions}
This work makes four main contributions:
\begin{enumerate}[noitemsep]
\item \textbf{Adaptation of classical reliability metrics to cognitive orchestration:}
We extend established dependability measures-MTTR, MTBF, and related ratios into the cognitive domain, defining MTTR-A as a runtime metric for reasoning recovery in multi-agent systems.
\item \textbf{A taxonomy of runtime recovery reflexes:}
We systematize reflexive control actions into five categories:-recovery, human-in-the-loop, runtime control, coordination, and safety-providing a standardized vocabulary for cognitive fault management and orchestration behavior.
\item \textbf{An empirical LangGraph simulation for cognitive reliability:}
We implement and release a LangGraph-based experiment (200 runs) that evaluates reasoning drift, recovery latency, and normalized reliability (NRR) across reflex modes.

\item \textbf{Formal reliability theorems:}
We derive two theoretical results linking runtime metrics to long-run cognitive uptime. 
First, under an alternating-renewal model \cite{barlow1965reliability}, we prove that $\mathrm{NRR}_{sys}$ provides a conservative lower bound 
on the steady-state cognitive uptime (Theorem~\ref{thm:nrr_bounds_uptime}). 
Second, we extend this relation to include recovery-time variance, establishing a confidence-aware lower bound 
that incorporates uncertainty in cognitive recovery (Theorem~\ref{thm:variance_bound}).

\end{enumerate}

\section{Related Work}
\textbf{LLM Agents and Deliberation:}
Chain-of-Thought, ReAct, Toolformer, and Tree-of-Thoughts \cite{yao2023tot} provide prompting, acting, and tool-use patterns for single agents. 
Beyond single-agent prompting, multi-agent lines-Generative Agents, CAMEL, Reflexion \cite{shinn2023reflexion}, and Voyager \cite{wang2023voyager}-together with recent surveys \cite{xi2023agentsurvey} map collaboration benefits and failure modes (e.g., coordination errors, unsafe tool calls, memory divergence).
These works primarily improve \emph{ex-ante} deliberation and coordination; they do not quantify \emph{ex-post} recovery latency or runtime dependability once reasoning has drifted.

Recent studies have expanded MAS collaboration paradigms
through reflection and debate mechanisms -such as Reflective Multi-Agent
Collaboration~\cite{copper2024reflective}, Chain-of-Agents~\cite{chainagents2024}, and MAGIS~\cite{magis2024}.

While these frameworks enhance coordination quality and communication efficiency,
they remain performance-oriented and do not formalize runtime dependability or
recovery-latency metrics as addressed in the present work.

\textbf{Resilient and Secure MAS:} Decades of MAS work study consensus, fault tolerance, and security: resilient consensus under DoS and Byzantine settings \cite{liu2023secureconsensus,wang2023resilientconsensus}, fault-tolerant control \cite{gao2022ftcc}, adaptive/secure consensus \cite{liu2023secureconsensus}, and broader resilient modelling \cite{varga2020resilient}. These works provide metrics, models, and control laws but do not standardize \emph{runtime recovery primitives} for LLM-centric agent orchestration.

These analyses reinforce the motivation for a quantitative dependability metric - directly measures recovery dynamics rather than task performance or interaction quality \cite{renkhoff2024survey}.

\vspace{0.4em}
Existing multi-agent AI systems such as Guardrails and AutoGen partially address reliability but operate at different layers. Guardrails enforces schema and safety validation \emph{after} text generation, while AutoGen structures multi-agent conversations and role interactions. Neither defines runtime recovery reflexes once coordination faults or reasoning drift occur. We introduce a control vocabulary that acts \emph{within the execution loop}-linking telemetry, policy, and reflex actions such as rollback or sandbox-execute. 

Recent multi-agent AI systems such as AutoGen, OpenAgents, and emerging control-plane studies including ALAS and Advancing Agentic Systems~\cite{adv_agentic_2024} extend LLM coordination to multi-agent workflows, but they primarily address communication and execution graphs rather than quantitative 
recovery metrics such as MTTR-A. 

Complementary benchmarking efforts evaluate multimodal professional workflow agents, yet omit explicit
reliability or recovery-time indicators, underscoring the absence of quantitative runtime dependability benchmarks in current agent evaluations.

In contrast to prior work on agent deliberation, coordination, or resilience, this paper focuses explicitly on recovery latency as a first-class reliability primitive for AI systems. Existing studies evaluate performance, robustness, or consensus properties, but do not define or measure the temporal dynamics of cognitive recovery following semantic or coordination failures. To our knowledge, MTTR-A is the first system-level metric that quantifies recovery time in distributed AI reasoning systems and links it formally to long-run cognitive uptime.

\section{Method}
This section formalizes the orchestration control method underlying the MTTR-A framework.
We propose a reflex-oriented architecture for runtime recovery in MAS,
where each cognitive fault triggers an automatic or human-mediated reflex action.
The architecture defines five families of recovery reflexes and organizes them within a closed feedback loop
linking telemetry, meta-monitoring, and execution (Figure~\ref{fig:control-loop-fixed}, Table~\ref{tab:taxonomy}).
Each recovery cycle produces measurable latency components from which MTTR-A
and related dependability metrics are computed (see Section~\ref{sec:mttra}).

\subsection{Reflex Taxonomy and Control Loop}
Each reflex family is defined by a \emph{trigger}, a \emph{policy}, and an \emph{effect on system state}.
Telemetry streams supply signals such as trust degradation, drift detection, or compliance violations,
which activate the appropriate reflex within the control loop.
The overall pipeline forms a reflex-based feedback system that continuously monitors cognition,
detects deviations, and enforces corrective actions.

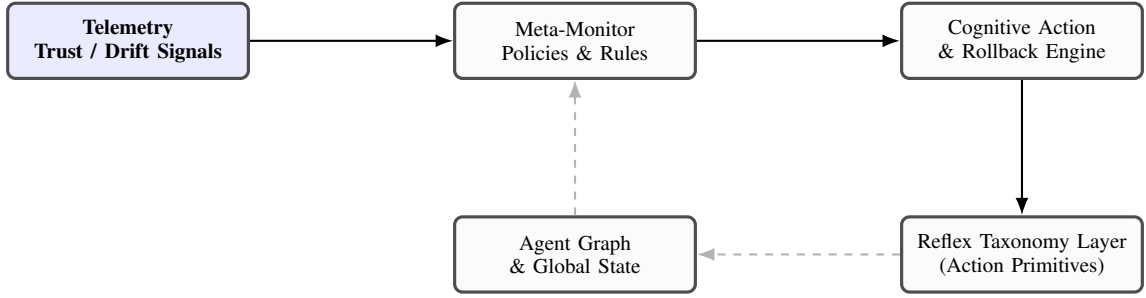
\begin{figure*}[t]
\centering
\begin{tikzpicture}[node distance=2.6cm,>=Latex,thick,every node/.style={font=\footnotesize,align=center}]

% --- Box style ---
\tikzstyle{blk}=[rectangle,rounded corners=3pt,draw=black!70,very thick,
minimum width=3.2cm,minimum height=1.0cm,fill=gray!3]

% --- Nodes (top row) ---
\node[blk,fill=blue!8,font=\bfseries\footnotesize] (tele) {Telemetry \\ Trust / Drift Signals};
\node[blk,right=2.7cm of tele] (monitor) {Meta-Monitor \\ Policies \& Rules};
\node[blk,right=2.7cm of monitor] (engine) {Cognitive Action \\ \& Rollback Engine};

% --- Nodes (bottom row) ---
\node[blk,below=1.8cm of monitor] (state) {Agent Graph \\ \& Global State};
\node[blk,below=1.8cm of engine] (actions) {Reflex Taxonomy Layer \\ (Action Primitives)};

% --- Arrows (main flow) ---
\draw[->,thick] (tele) -- (monitor);
\draw[->,thick] (monitor) -- (engine);
\draw[->,thick] (engine) -- (actions);

% --- Feedback arrows ---
\draw[->,dashed,gray!60] (actions.west) -- (state.east);
\draw[->,dashed,gray!60] (state.north) -- (monitor.south);

\end{tikzpicture}
\caption{
Closed-loop control architecture linking telemetry, policy monitoring, action execution,
and global state feedback. Arrows represent the data flow between perception (telemetry),
decision (meta-monitor), and action (execution engine),
forming a reflex-based feedback system for the Multi-Agent System (MAS).
}
\label{fig:control-loop-fixed}
\end{figure*}

\begin{table*}[h!]
\centering
\caption{Reflex taxonomy and corresponding recovery actions.
Each category defines runtime mechanisms for detection, intervention,
and safe continuation within the orchestration control loop of the MAS.}
\label{tab:taxonomy}
\small
\renewcommand{\arraystretch}{1.25}
\setlength{\tabcolsep}{6pt}
\begin{tabular}{@{}p{2.5cm} p{3.2cm} p{4.2cm} p{4.5cm}@{}}
\toprule
\textbf{Category} & \textbf{Action} & \textbf{Purpose} & \textbf{Typical Trigger} \\ \midrule

\textbf{Recovery} &
auto-replan, rollback, tool-retry, fallback-policy, safe-mode &
Regenerate plans, restore checkpoints, or retry failed tools using simplified policies and minimal verified behaviors. &
Planning failure, drift event, tool timeout, repeated failure, or safety constraint. \\[3pt]

\textbf{Human-in-the-Loop} &
human-approve, human-override, human-review, escalate-to-expert &
Introduce manual oversight, approval, or escalation to experts for compliance-critical or ambiguous cases. &
Compliance or high-impact step, anomaly detected, or specialized task requiring human judgment. \\[3pt]

\textbf{runtime control} &
auto-diagnose, self-heal, confidence-gate, vote-or-consensus, sandbox-execute, drift-rollback &
Perform self-diagnosis, pause and validate, isolate risky actions, or revert models when drift or uncertainty is detected. &
Anomaly pattern, recoverable environment fault, low confidence, divergent outcomes, or drift threshold exceeded. \\[3pt]

\textbf{Coordination} &
broadcast-update, negotiate-task, sync-state, lock/release-resource &
Maintain distributed coherence across the MAS through shared context propagation, task negotiation, and resource synchronization. &
Global context change, workload imbalance, context divergence, or resource contention. \\[3pt]

\textbf{Safety} &
graceful-abort, force-terminate, audit-snapshot &
Ensure controlled or emergency shutdown with full trace capture for audit and recovery verification. &
Controlled shutdown condition, unrecoverable error, or compliance requirement. \\

\bottomrule
\end{tabular}
\vspace{0.5em}
\end{table*}

\paragraph{Behavioral Scope:}
The reflex families in Table~\ref{tab:taxonomy} define complementary layers of recovery control, from automatic replanning to human-in-the-loop oversight and safety shutdown. Collectively, they ensure that cognitive drift, tool faults, and coordination errors are corrected without systemic collapse.

\subsection{Runtime Policy and Coordination Dynamics}
Within a MAS, recovery latency is shaped by the type of reflex invoked and also by the efficiency of policy selection and inter-agent coordination.
The meta-monitor chooses the fastest valid recovery action based on telemetry triggers, such as tool errors, confidence loss, or detected drift, while policy compliance.
Agents coordinate through a shared global context, synchronizing memory and actions during recovery to prevent cascading inconsistencies.
These decision and communication latencies collectively define the measured MTTR-A: the time interval from fault detection to complete restoration of coherent operation across the MAS.

The decision process that governs reflex selection maps fault triggers-such as tool errors, low confidence, or reasoning drift-to 
their corresponding recovery actions within the control loop.

Reducing MTTR-A therefore depends on minimizing both the policy-selection delay and the coordination overhead that emerge during distributed recovery.

\section{Experiment and Evaluation Metrics}
\label{sec:mttra}

This section presents both the theoretical framework and the experimental configuration 
used to evaluate the proposed runtime reliability metrics for MAS. 
First, we establish the conceptual motivation and formal definitions of MTTR-A and related measures that quantify cognitive recovery and dependability. 
Then, we describe the experimental setup and data pipeline used to empirically compute these metrics 
in a simulated MAS built with the LangGraph framework.

\subsection{Evaluation Metrics and Cognitive Dependability Context}
The proposed reliability metrics were evaluated using the LangGraph framework, 
which enables stateful, graph-based MAS workflows built on LLMs. The experimental setup simulated reasoning drift, reflex activation, and recovery cycles 
across multiple interacting agents, allowing quantitative assessment of runtime cognitive dependability.

We draw directly from classical dependability theory \cite{bass2020sre}, 
which formalizes reliability using MTTR, MTBF, and availability ratios. 
In classical terms, MTTR measures the expected duration between system failure 
and full restoration of service under a binary up-down model. 
In contrast, cognitive systems such as MAS operate continuously, with faults that are semantic, 
distributed, and partially recoverable. 
We therefore reinterpret these dependability metrics at the level of cognitive orchestration, quantifying the efficiency with which a MAS detects, isolates, and restores reasoning coherence following a coordination or drift event.

\subsection{Formal Reliability Metricization Framework}
We adapt classical dependability measures to the MAS domain, 
defining \emph{MTTR-A} (Mean Time-to-Recovery for Agentic Systems) 
as a runtime measure of cognitive recovery latency. 
While traditional MTTR quantifies infrastructural repair time, 
MTTR-A captures the temporal efficiency of \emph{cognitive recovery}-the process 
by which a distributed reasoning network identifies drift and restores coherent operation.

\paragraph{System-of-Agents Model:}
We model the orchestrated workflow as a \textbf{system of agents} 
$\mathcal{S} = \{\mathcal{A}_1, \mathcal{A}_2, \ldots, \mathcal{A}_N\}$, 
where each agent $\mathcal{A}_i$ functions as an autonomous \emph{subsystem} 
with its own local reasoning, detection, and recovery dynamics. 
Each subsystem executes internal reflex components $c_{ik}$ (e.g., 
\texttt{tool-retry}, \texttt{auto-replan}, \texttt{rollback}, \texttt{human-approve}) 
that act as \emph{components} of the overall recovery architecture. 
Thus, MAS reliability is described hierarchically:
\begin{itemize}[noitemsep]
    \item \textbf{Component level:} reflex-mode recovery latency ($MTTR\text{-}A_{\text{mode}}$)
    \item \textbf{Subsystem level:} per-agent metrics ($MTTR\text{-}A_i$, $MTBF_i$)
    \item \textbf{System level:} aggregate orchestration metrics ($MTTR\text{-}A_{sys}$, $MTBF_{sys}$, $NRR_{sys}$)
\end{itemize}
This hierarchy mirrors systems engineering principles, where each subsystem contributes 
independently to overall reliability.

\paragraph{Formal Definition:}
Let $\{(t_{f,i}^{(j)}, t_{r,i}^{(j)})\}_{j=1}^{M_i}$ denote the set of 
$M_i$ fault-recovery episodes for agent $\mathcal{A}_i$, 
where $t_{f,i}^{(j)}$ is the timestamp at which a cognitive fault is detected 
and $t_{r,i}^{(j)}$ is when reasoning coherence is re-established. 
The local recovery duration is:
\begin{equation}
\Delta T_i^{(j)} = t_{r,i}^{(j)} - t_{f,i}^{(j)}.
\label{eq:deltaT}
\end{equation}
The per-agent \textbf{Mean Time-To-Recovery} is then:
\begin{equation}
\mathrm{MTTR\text{-}A}_i = 
\frac{1}{M_i}\sum_{j=1}^{M_i}\!\left(t_{r,i}^{(j)} - t_{f,i}^{(j)}\right)
= \mathbb{E}[\Delta T_i],
\label{eq:mttra_mean}
\end{equation}
and the system-level value, aggregating across $N$ agents, is defined as:
\begin{equation}
\mathrm{MTTR\text{-}A}_{sys} = 
\frac{1}{N}\sum_{i=1}^{N}\mathrm{MTTR\text{-}A}_i.
\label{eq:mttra_sys}
\end{equation}

\paragraph{Robust Estimator:}
As recovery-time distributions in MAS are typically right-skewed 
due to occasional long-latency interventions (e.g., human approvals), 
we also suggest using a robust estimator, the \textbf{Median Time-To-Recovery}:
\begin{equation}
\mathrm{MedTTR\text{-}A_{sys}} =
\operatorname{median}_{1 \le j \le M}\!\left[\Delta T^{(j)}\right].
\label{eq:mttra_median}
\end{equation}
This mitigates the influence of rare but extreme events 
and complements MTTR-A as a distribution-invariant measure.

\paragraph{Latency Decomposition:}
Each recovery episode, at either agent or system level, can be decomposed into additive latency components:
\begin{equation}
\Delta T^{(j)} = 
T_{\mathrm{detect}}^{(j)} + 
T_{\mathrm{decide}}^{(j)} + 
T_{\mathrm{execute}}^{(j)},
\label{eq:decomposition}
\end{equation}
where $T_{\mathrm{detect}}$ is the drift detection latency, 
$T_{\mathrm{decide}}$ is the policy-selection delay, 
and $T_{\mathrm{execute}}$ is the reflex execution time.
\begin{figure*}[h!]
\centering
\begin{tikzpicture}[
  >=Latex,
  every node/.style={font=\scriptsize, align=center},
  box/.style={rectangle, rounded corners=2pt, draw=gray!70, thick,
              fill=gray!3, minimum width=1.8cm, minimum height=0.7cm}
]

% --- Nodes ---
\node[box, fill=gray!10] (fault) {Drift\\Detected};
\node[box, fill=blue!10, right=1.7cm of fault] (policy) {Trigger\\Policy};
\node[box, fill=purple!10, right=1.7cm of policy] (recovery) {Recovery\\Action};
\node[box, fill=gray!10, right=1.7cm of recovery] (stable) {Stable\\State};

% --- Arrows (labels placed separately) ---
\draw[->, thick, blue!70] (fault) -- (policy);
\node[font=\tiny, text=blue!70, above=2pt of $(fault)!0.5!(policy)$] {Detection};

\draw[->, thick, purple!70] (policy) -- (recovery);
\node[font=\tiny, text=purple!70, above=2pt of $(policy)!0.5!(recovery)$] {Decision};

\draw[->, thick, gray!70] (recovery) -- (stable);
\node[font=\tiny, text=gray!70, above=2pt of $(recovery)!0.5!(stable)$] {Execution};

% --- MTTR-A label and horizontal arrow ---
\node[font=\scriptsize\bfseries, text=blue!70, left=0.35cm of fault] (label) {MTTR-A};
\draw[<->, thick, blue!50!purple!60] 
  ([yshift=-0.35cm]fault.south west) -- ([yshift=-0.35cm]stable.south east);

\end{tikzpicture}

\caption{Compact MTTR-A timeline showing detection (blue), decision (purple), and execution (gray) phases leading to recovery. The horizontal arrow below represents the total MTTR-A duration.}
\label{fig:mttra-timeline}
\end{figure*}
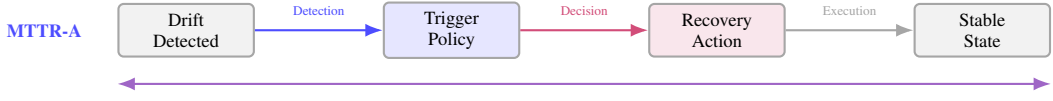

\paragraph{Complementary Metrics:}
Two additional dependability metrics complete the MAS reliability framework:
\begin{enumerate}[noitemsep]
\item \textbf{Mean Time Between Cognitive Faults (MTBF):} 
average stable duration between drift events:
\begin{equation}
\mathrm{MTBF}_i =
\frac{1}{M_i}
\sum_{j=1}^{M_i}
\!\left(
t^{(j)}_{\mathrm{fault},i} -
t^{(j-1)}_{\mathrm{recovered},i}
\right),
\label{eq:mtbf}
\end{equation}
and the system-level MTBF:
\begin{equation}
\mathrm{MTBF}_{sys} =
\frac{1}{N}\sum_{i=1}^{N}\mathrm{MTBF}_i.
\label{eq:mtbf_sys}
\end{equation}
\item \textbf{NRR:} 
a dimensionless runtime reliability index:
\begin{equation}
\mathrm{NRR}_{sys} = 1 - 
\frac{\mathrm{MTTR\text{-}A}_{sys}}{\mathrm{MTBF}_{sys}}.
\label{eq:nrr}
\end{equation}
Values near $1$ indicate MAS that recover rapidly relative to fault rate, 
while values approaching $0$ or negative imply unstable recovery behavior.
\end{enumerate}

We adapt the classical MTBF-MTTR relationship into a normalized,
runtime reliability indicator for cognitive orchestration/ MAS,
termed the \emph{Normalized Recovery Ratio (NRR)}.
While the underlying form derives from traditional availability approximations, its application here is newly adopted: the NRR is directly measurable from reasoning-drift and recovery events in MAS, providing a practical indicator of cognitive uptime during runtime operation.

We next formalize this relationship under the alternating-renewal model of dependability theory \cite{bass2020sre}, showing that $\mathrm{NRR}_{sys}$ acts as a
conservative lower bound on the system’s long-run cognitive uptime.

\begin{definition}[Steady-State Fraction of Cognitive Uptime]
\label{def:steady_state}
Let the MAS operate as an alternating sequence of 
cognitively stable (\emph{up}) and recovery (\emph{down}) periods over time.
Denote by $A_{\mathrm{up}}(T)$ the total duration of coherent reasoning 
within the time interval $[0, T]$. 
The \emph{steady-state fraction of cognitive uptime} is defined as the long-run ratio
\begin{equation}
\pi_{\mathrm{up},sys} := 
\lim_{T \to \infty} \frac{A_{\mathrm{up}}(T)}{T},
\label{eq:steady_state_def}
\end{equation}
whenever the limit exists. 
It represents the asymptotic proportion of time during which the MAS 
maintains reasoning coherence, analogous to classical system availability.
\end{definition}

\begin{theorem}[NRR lower-bounds steady-state cognitive uptime]
\label{thm:nrr_bounds_uptime}
Under the alternating-renewal model of dependability theory,
the steady-state fraction of cognitive uptime in a MAS satisfies
\begin{equation}
\begin{aligned}
\pi_{\mathrm{up},sys}
&= \frac{\mathrm{MTBF}_{sys}}
        {\mathrm{MTBF}_{sys}+\mathrm{MTTR}_{sys}}
 = \frac{1}{1+\lambda_{sys}\mu_{sys}} \\
&\ge 1 - \lambda_{sys}\mu_{sys}
 = \mathrm{NRR}_{sys}.
\end{aligned}
\end{equation}

where $\lambda_{sys} = 1/\mathrm{MTBF}_{sys}$ 
and $\mu_{sys} = \mathrm{MTTR}_{sys}$.
Therefore, the normalized recovery ratio $\mathrm{NRR}_{sys}$ 
provides a conservative lower bound on the system’s 
steady-state cognitive uptime.
The inequality follows from the elementary bound 
$(1+x)^{-1} \ge 1-x$ for all $x \ge 0$.
A proof is provided in Appendix~\ref{app:proofs}.
\end{theorem}

\begin{definition}[Confidence-aware Normalized Recovery Ratio ($\mathrm{NRR}_\alpha$)]
Let $R$ denote the random recovery duration of the MAS, with mean 
$\mu = \mathbb{E}[R] = \mathrm{MTTR\text{-}A}_{sys}$ and standard deviation 
$\sigma = \sqrt{\mathrm{Var}[R]}$. 
For any confidence level $\alpha \in (0,1)$, define 
\[
k_\alpha = \sqrt{\frac{1-\alpha}{\alpha}}, 
\qquad
R_\alpha := \mu + k_\alpha \sigma.
\]
Then, the \emph{confidence-aware normalized recovery ratio} is
\begin{equation}
\mathrm{NRR}_\alpha := 1 - \lambda_{sys} R_\alpha,
\qquad \text{where } \lambda_{sys} = 1/\mathrm{MTBF}_{sys}.
\end{equation}
This extends the classical $\mathrm{NRR}_{sys}$ by explicitly incorporating 
the variance of recovery times, providing a confidence-graded lower bound on 
runtime cognitive uptime.
\end{definition}

\begin{theorem}[Variance-aware lower bound on cognitive uptime]
\label{thm:variance_bound}
Under the same alternating-renewal model of dependability theory,
for any confidence level $\alpha \in (0,1)$ the steady-state fraction of 
cognitive uptime satisfies
\begin{equation}
\pi_{\mathrm{up},sys} \ge 1 - \lambda_{sys}(\mu + k_\alpha \sigma)
= \mathrm{NRR}_\alpha,
\end{equation}
with probability at least $\alpha$.
A complete proof is provided in Appendix~\ref{app:proof_variance_bound}.
\end{theorem}
\begin{figure}[h!]
\centering
\includegraphics[width=0.9\linewidth]{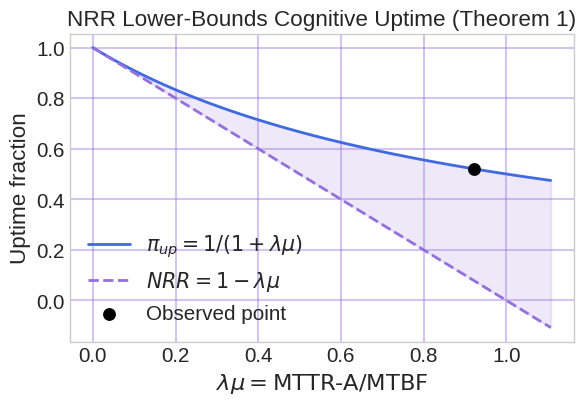}
\caption{
\textbf{NRR and variance-aware bounds on cognitive uptime.}
The blue curve shows the analytical steady-state uptime 
$\pi_{\mathrm{up},sys} = 1/(1+\lambda_{sys}\mu_{sys})$,
the dashed purple line shows the first-order approximation 
$\mathrm{NRR}_{sys} = 1-\lambda_{sys}\mu_{sys}$,
and the gray shaded band represents the confidence-adjusted bounds 
$\mathrm{NRR}_\alpha$ from Theorem~\ref{thm:variance_bound}.
}

\label{fig:nrr-uptime}
\end{figure}

The adapted metrics 
$\mathrm{MTTR\text{-}A}_i$, $\mathrm{MedTTR\text{-}A}_i$, $\mathrm{MTBF}_i$, and $\mathrm{NRR}_i$ 
(for each agent), and their system-level counterparts 
$\mathrm{MTTR\text{-}A}_{sys}$, $\mathrm{MTBF}_{sys}$, and $\mathrm{NRR}_{sys}$, 
establish a unified quantitative foundation for evaluating 
\emph{runtime cognitive dependability} in MAS.

\subsection{Experimental Setup}
The experimental evaluation employed LangGraph to simulate cognitive drift, 
reflex activation, and recovery across interacting agents. 
The following configuration was used.

\subsection{Algorithmic Procedure for Measuring MTTR-A}
Algorithm~\ref{alg:mttra} defines a general procedure for measuring
cognitive reliability metrics in multi-agent systems.
It is agnostic to the underlying framework or dataset and can be applied 
to any orchestrated workflow that supports fault detection and recovery tracing.
In our experiments, this procedure was instantiated using the LangGraph
framework and the AG~News corpus to simulate reasoning-drift-recovery cycles.

\begin{algorithm}[h!]
\caption{General Procedure for Computing Cognitive Reliability Metrics in Multi-Agent Systems}
\label{alg:mttra}
\begin{algorithmic}[1]
\Require
Multi-agent system $\mathcal{S}$, fault detection policy $\pi_{\mathrm{detect}}$, 
recovery reflex set $\mathcal{R}$, number of iterations $N$.
\Ensure
System-level $\mathrm{MTTR\text{-}A}_{sys}$, $\mathrm{MTBF}_{sys}$, $\mathrm{NRR}_{sys}$.

\For{$i = 1$ to $N$}
    \State Observe system state; apply $\pi_{\mathrm{detect}}$ to determine fault occurrence.
    \If{fault detected}
        \State Trigger appropriate recovery reflex $r \in \mathcal{R}$
        \State Measure detection, decision, and execution latencies
        \State Compute total recovery time $\Delta T_i = 
               T_{\mathrm{detect}} + T_{\mathrm{decide}} + T_{\mathrm{execute}}$
    \EndIf
\EndFor
\State Estimate $\mathrm{MTTR\text{-}A}_{sys} = \operatorname{median}(\Delta T_i)$
\State Estimate $\mathrm{MTBF}_{sys}$ from inter-fault intervals
\State Compute $\mathrm{NRR}_{sys} = 1 - 
      \dfrac{\mathrm{MTTR\text{-}A}_{sys}}{\mathrm{MTBF}_{sys}}$
\State \Return $\{\mathrm{MTTR\text{-}A}_{sys}, \mathrm{MTBF}_{sys}, \mathrm{NRR}_{sys}\}$
\end{algorithmic}
\end{algorithm}
\noindent
In our implementation, Algorithm~\ref{alg:mttra} was instantiated in the 
LangGraph multi-agent AI system as a three-node workflow representing the
reasoning-drift-recovery cycle. The nodes operated sequentially as follows:

\textit{(1)}~\texttt{reasoning\_node} retrieved documents from the AG~News corpus
and computed a retrieval confidence score,
\begin{equation}
c = \cos(\mathbf{q}, \mathbf{d}_{\mathrm{top}}) 
= \frac{\mathbf{q} \cdot \mathbf{d}_{\mathrm{top}}}
{\|\mathbf{q}\|\|\mathbf{d}_{\mathrm{top}}\|},
\label{eq:confidence}
\end{equation}
where $\mathbf{q}$ is the query vector and $\mathbf{d}_{\mathrm{top}}$ is the
highest-scoring document. 

\textit{(2)}~\texttt{check\_drift\_node} compared $c$ to the drift threshold
$\tau_{\mathrm{drift}} = 0.6$ and flagged a cognitive fault whenever 
$c < \tau_{\mathrm{drift}}$ (or via small stochastic perturbations).  

\textit{(3)}~\texttt{recovery\_node} executed one of four reflex modes-
\texttt{auto-replan}, \texttt{rollback}, \texttt{tool-retry}, or
\texttt{human-approve}-selected according to weighted policy sampling and
simulated decision and execution latencies.

All state transitions were logged as structured telemetry (\texttt{.jsonl})
for offline computation of $\mathrm{MTTR\text{-}A}$, $\mathrm{MTBF}$, and
$\mathrm{NRR}$.  
Across 200 runs, timestamps for detection, decision, and execution were recorded,
yielding per-mode recovery latencies and system-level dependability metrics.
\subsection{Interpretation}
This simulation reproduces distributed reasoning scenarios 
where autonomous agents collaborate under uncertainty. 
In such environments, $MTTR-A_{sys}$ captures the MAS’s operational resilience, how rapidly reasoning coherence is restored after drift, providing a practical runtime indicator of cognitive reliability. Further implementation details, including computational environment, query pool, and configuration parameters, are provided in Appendix~\ref{app:Configuration}.

\section{Results and Discussion}
The overall MedTTR-A across 200 LangGraph runs on the AG News dataset was
6.21~s~$\pm$~2.14~s (90th percentile: 10.73~s),
indicating that typical cognitive recovery occurs within roughly six seconds after drift detection. The MTBF was
6.73~s~$\pm$~2.14~s, yielding an NRR of 0.077.
These results confirm that the orchestration maintained stable recovery cycles with
approximately 8\% of runtime devoted to corrective action and no evidence of cascading
instability.

These metrics indicate that the agentic workflow maintained stable recovery 
cycles with roughly 10\% of runtime spent on corrective actions and no simulated evidence of cascading instability. Per-mode results are summarized in Table~\ref{tab:mttra-summary}.

\begin{table}[h!]
\centering
\caption{Experimental results on the AG News dataset (N = 200 runs).
Each value represents the median recovery time per reflex mode.}
\label{tab:mttra-summary}
\small
\begin{tabular}{@{}lcccc@{}}
\toprule
\textbf{Recovery Mode} & \textbf{MedTTR-A (s)} & \textbf{Std (s)} & \textbf{P90 (s)} & \textbf{Count} \\ \midrule
auto-replan     & 5.94 & 0.70 & 6.81 & 93 \\
tool-retry      & 4.46 & 0.61 & 5.40 & 42 \\
rollback        & 6.99 & 0.43 & 7.45 & 44 \\
human-approve   & 12.22 & 0.68 & 12.77 & 21 \\ \bottomrule
\end{tabular}
\end{table}

\subsection*{Distribution and Recovery Behavior}
Figure \ref{fig:mttra-dist-box} summarizes both the overall distribution and per-mode latency comparison. The results follow a right-skewed profile centered near 6-7~s, 
with most automated recoveries completing rapidly and a small tail corresponding to human oversight. 
This shape reflects high orchestration responsiveness and low variance across the experiment. Fully automated reflexes (\texttt{tool-retry}, \texttt{auto-replan}) restored stability within 5-7~s on average,
while \texttt{rollback} operations showed slightly higher median times but reduced variability.
The \texttt{human-approve} mode remained the slowest due to manual gating.

\begin{figure}[h!]
\centering
\includegraphics[width=0.48\linewidth]{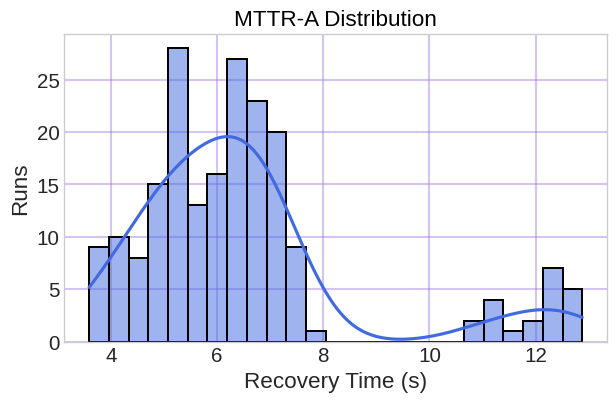}
\includegraphics[width=0.48\linewidth]{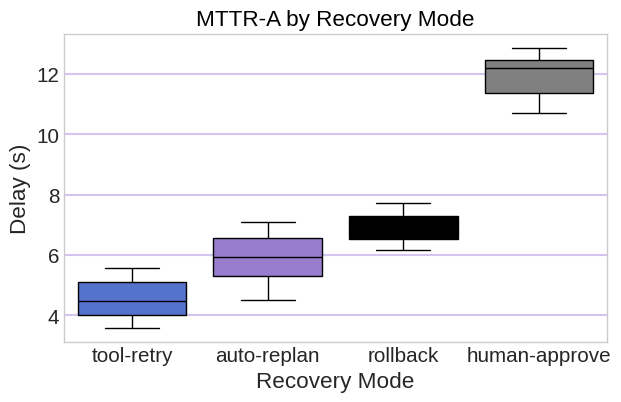}
\caption{
\textbf{Distribution and Mode Comparison of $MTTR\text{-}A_{sys}$.}
Left: overall histogram and kernel density of recovery latency across 200 runs.
Right: boxplots per reflex strategy (\texttt{tool-retry}, \texttt{auto-replan}, 
\texttt{rollback}, \texttt{human-approve}).}
\label{fig:mttra-dist-box}
\end{figure}

\subsection*{Temporal Stability Across Runs}
To assess temporal stability, we analyzed the rolling 20-run MTTR-A average over the 200 experimental iterations. The results indicate steady recovery performance without systematic escalation, suggesting stable reflex control and consistent orchestration efficiency throughout the experiment. Minor fluctuations are attributable to stochastic variation in recovery mode selection rather than degradation in system behavior.

To analyze recovery behavior, MTTR-A was decomposed into three additive components: detection, decision, and execution latency. Across all reflex modes, execution time constitutes the dominant contribution to recovery cost, with the effect most pronounced in human-approval cases. Decision latency remains consistently low, indicating that policy selection overhead within the LangGraph orchestration layer is negligible relative to reflex execution time.

Fully automated reflexes, such as auto-replan and tool-retry, restore stability within approximately 5–7 s with low variance, supporting high-throughput operation. Human-in-the-loop interventions introduce expected runtime control latency but provide necessary compliance and control in high-risk scenarios.

\subsection*{Comparison to Existing Metrics}

In classical dependability research \cite{bass2020sre}, 
the MTTR quantifies the average duration between a system failure 
and full service restoration.  
This formulation assumes binary operational states and deterministic repair procedures. We adapt these classical notions to \emph{agentic cognition}, 
where recovery refers to reasoning correction, consensus restoration, or human oversight 
rather than physical service uptime. Traditional MTTR and MTBF describe infrastructure recovery, not cognitive correction.  
MTTR-A therefore complements-not replaces-these metrics by quantifying recovery latency 
within multi-agent reasoning loops, bridging classical runtime monitoring with cognitive control performance.

In LLM-based agents, quality depends on plan consistency, tool correctness, and collaborative coherence.  
Methods such as ReAct and Tree-of-Thoughts improve \emph{ex-ante} reasoning 
\cite{yao2022react,yao2023tot}, whereas our approach focuses on \emph{ex-post} reflexes-actions that 
restore safe state, replan, or escalate after reasoning divergence.  
From MAS literature, resilient consensus ensures state agreement under faults or attacks, 
but does not model human-approval gates, sandbox execution, or audit snapshots required in regulated operational workflows.

MTTR-A can serve as a building block for higher-level aggregate indicators, such as 
a \emph{Mean Cognitive Uptime (MCU)} or a \emph{Cognitive Reliability Index (CRI)}-which define measurable 
targets for adaptive AI performance and runtime control.  

\section*{Limitations}

The present evaluation relies on a semi-simulated environment with a synthetic corpus and
controlled reflex latencies, isolating runtime dependability from semantic variability. Future validation should include open-world orchestration traces (for example, customer-service
and data-analysis agents) to test MTTR-A under heterogeneous latency and partial runtime monitoring.
Integrating semantic accuracy metrics such as task-success, BLEU, or F1 alongside MTTR-A
would allow joint evaluation of reasoning quality and recovery efficiency, establishing practical
trade-offs between cognitive performance and dependability.

\section{Conclusion and Future Work}
This work adapts classical reliability principles to the cognitive domain of multi-agent systems,
proposing a practical framework for evaluating \textbf{runtime cognitive dependability}.
By extending traditional dependability metrics-MTTR, MTBF, and availability-into the cognitive
context, we defined \textbf{MTTR-A} as a measurable indicator of reasoning recovery latency
and introduced the complementary metrics \textbf{MTBF} and \textbf{NRR} for continuous
assessment of orchestration stability.

The proposed framework integrates a taxonomy of reflexive control actions with empirical
evaluation in LangGraph, establishing a reproducible basis for quantifying how distributed
agentic systems detect, recover, and stabilize after reasoning drift or coordination faults.
This connection between fault-tolerance engineering and cognitive orchestration lays
the groundwork for standardized runtime stability benchmarks across LLM-based multi-agent
systems.

Future work will extend this foundation to real-world heterogeneous teams, integrating
semantic performance metrics with runtime dependability measures to assess trade-offs
between reasoning quality and recovery efficiency. Further directions include exploring
runtime control-aware reliability objectives, safety guarantees for reflexive control loops,
and alignment with emerging standards for operational AI reliability. In practice, MTTR-A enables organizations to define service-level objectives (SLOs) for agentic systems based on recovery behavior rather than task accuracy alone.

\appendix
\numberwithin{equation}{section}

\section{Proofs}
\label{app:proofs}

\begin{proof}[Proof of Theorem~\ref{thm:nrr_bounds_uptime}]
Consider the alternating-renewal process $\{(U_n,D_n)\}_{n\ge 1}$ for the MAS,
with $U_n$ (cognitive up-time) and $D_n$ (recovery down-time) i.i.d. and
$\mathbb{E}[U]=\mathrm{MTBF}_{sys}<\infty$, $\mathbb{E}[D]=\mathrm{MTTR}_{sys}<\infty$.
Define $\lambda_{sys}:=1/\mathrm{MTBF}_{sys}$ and $\mu_{sys}:=\mathrm{MTTR}_{sys}$.

By the Renewal-Reward Theorem, the steady-state fraction of cognitive up-time is
\begin{equation}
\begin{aligned}
\pi_{\mathrm{up},sys}
&= \frac{\mathbb{E}[U]}
        {\mathbb{E}[U]+\mathbb{E}[D]}
 = \frac{\mathrm{MTBF}_{sys}}
        {\mathrm{MTBF}_{sys}+\mathrm{MTTR}_{sys}} \\
&= \frac{1}{1+\lambda_{sys}\mu_{sys}}.
\end{aligned}
\label{eq:renewal-uptime}
\end{equation}

The normalized recovery ratio is defined by
\begin{equation}
\mathrm{NRR}_{sys} := 1-\lambda_{sys}\mu_{sys}.
\label{eq:nrr-def}
\end{equation}
Let $a:=\lambda_{sys}\mu_{sys}\ge 0$. Then
\begin{equation}
\frac{1}{1+a}-(1-a)
= \frac{a^2}{1+a} \;\ge\; 0,
\label{eq:gap-nonneg}
\end{equation}
with equality iff $a=0$. Combining \eqref{eq:renewal-uptime}-\eqref{eq:gap-nonneg} yields
\begin{equation}
\pi_{\mathrm{up},sys}
= \frac{1}{1+\lambda_{sys}\mu_{sys}}
\;\ge\;
1-\lambda_{sys}\mu_{sys}
= \mathrm{NRR}_{sys}.
\label{eq:uptime-lb}
\end{equation}
Moreover, the first-order expansion $(1+a)^{-1}=1-a+\mathcal{O}(a^{2})$ implies
\begin{equation}
\pi_{\mathrm{up},sys}
= \mathrm{NRR}_{sys} + \mathcal{O}\!\left((\lambda_{sys}\mu_{sys})^{2}\right),
\quad \text{as } \lambda_{sys}\mu_{sys}\to 0.
\label{eq:first-order}
\end{equation}
This proves that $\mathrm{NRR}_{sys}$ is a conservative lower bound on steady-state
cognitive up-time and a first-order approximation thereof.
\end{proof}

\begin{proof}[Proof of Theorem~\ref{thm:variance_bound}]
\label{app:proof_variance_bound}
Let $R$ denote the random recovery duration of a single cognitive-fault episode. 
By Cantelli’s one-sided inequality, for any $k>0$,
\begin{equation}
\Pr[R - \mu \ge k\sigma] \le \frac{1}{1 + k^2},
\label{eq:cantelli_app}
\end{equation}
where $\mu=\mathbb{E}[R]$ and $\sigma=\sqrt{\mathrm{Var}[R]}$.
Setting $k = k_\alpha = \sqrt{(1-\alpha)/\alpha}$ yields
\[
\Pr[R \le \mu + k_\alpha \sigma] \ge \alpha.
\]
Define $R_\alpha = \mu + k_\alpha \sigma$. 
With probability at least $\alpha$, the recovery duration satisfies $R \le R_\alpha$.

The steady-state fraction of cognitive uptime for an alternating-renewal process is
\begin{equation}
\pi_{\mathrm{up},sys} = \frac{1}{1 + \lambda_{sys}R},
\label{eq:pi_appendix}
\end{equation}
where $\lambda_{sys}=1/\mathrm{MTBF}_{sys}$ is the fault intensity.
Since $(1+x)^{-1}$ is monotonically decreasing for $x>0$, 
inequality \eqref{eq:cantelli_app} implies
\begin{equation}
\pi_{\mathrm{up},sys} \ge \frac{1}{1+\lambda_{sys}R_\alpha},
\quad \text{with probability at least } \alpha.
\end{equation}
Finally, applying the elementary inequality $(1+x)^{-1} \ge 1-x$ for $x\ge0$ gives
\begin{equation}
\pi_{\mathrm{up},sys} 
\ge 1 - \lambda_{sys} R_\alpha
= 1 - \lambda_{sys}(\mu + k_\alpha \sigma)
= \mathrm{NRR}_\alpha.
\end{equation}
Thus, with probability at least $\alpha$, the steady-state cognitive uptime 
exceeds the confidence-aware lower bound $\mathrm{NRR}_\alpha$, completing the proof.
\end{proof}

\section{Query Pool}
\label{app:querypool}

For completeness, the following query pool was used to induce reasoning drift 
and test recovery reflexes within the LangGraph experimental workflow. 
These orchestration-level queries represent typical coordination, runtime control, 
and recovery scenarios observed in agentic systems.

\begin{lstlisting}
# === Query Pool ===
QUERY_POOL = [
  "LangGraph recovery reflexes",
  "agent orchestration reliability",
  "rollback sandbox audit snapshots",
  "tool retries and backoff",
  "consensus voting disagreement",
  "policy thresholds and approvals",
  "runtime control and risk tiers",
  "runtime monitoring telemetry signals",
  "drift detection and confidence",
  "mttra and mtbf calculation",
  "normalized reliability index",
  "incident playbooks escalation",
  "global memory reconciliation",
  "safe mode fallback routes",
  "retrieval reranking grounding",
  "planning decomposition tools"
]
\end{lstlisting}

Each query was submitted through the reasoning-drift-recovery cycle defined in 
Section~\ref{sec:mttra}, ensuring consistent cognitive fault induction 
and enabling reproducible measurement of recovery latency and MTTR-A dynamics.

\section{Experimental Configuration}
\label{app:Configuration}
All experiments were executed on Google Colab with an NVIDIA A100 GPU, 
using Python~3.10 and the LangGraph~0.0.52 library. 
Each run instantiated a stateful graph of three nodes 
(\emph{Reasoning}, \emph{Drift Check}, and \emph{Recovery}), 
executed 200 independent iterations with random seeds fixed to ensure reproducibility. 

% ===================== References =====================
\small
\bibliographystyle{IEEEtran}

\bibliography{references}

\begin{IEEEbiography}[{\includegraphics[width=1in,height=1.25in,clip,keepaspectratio]{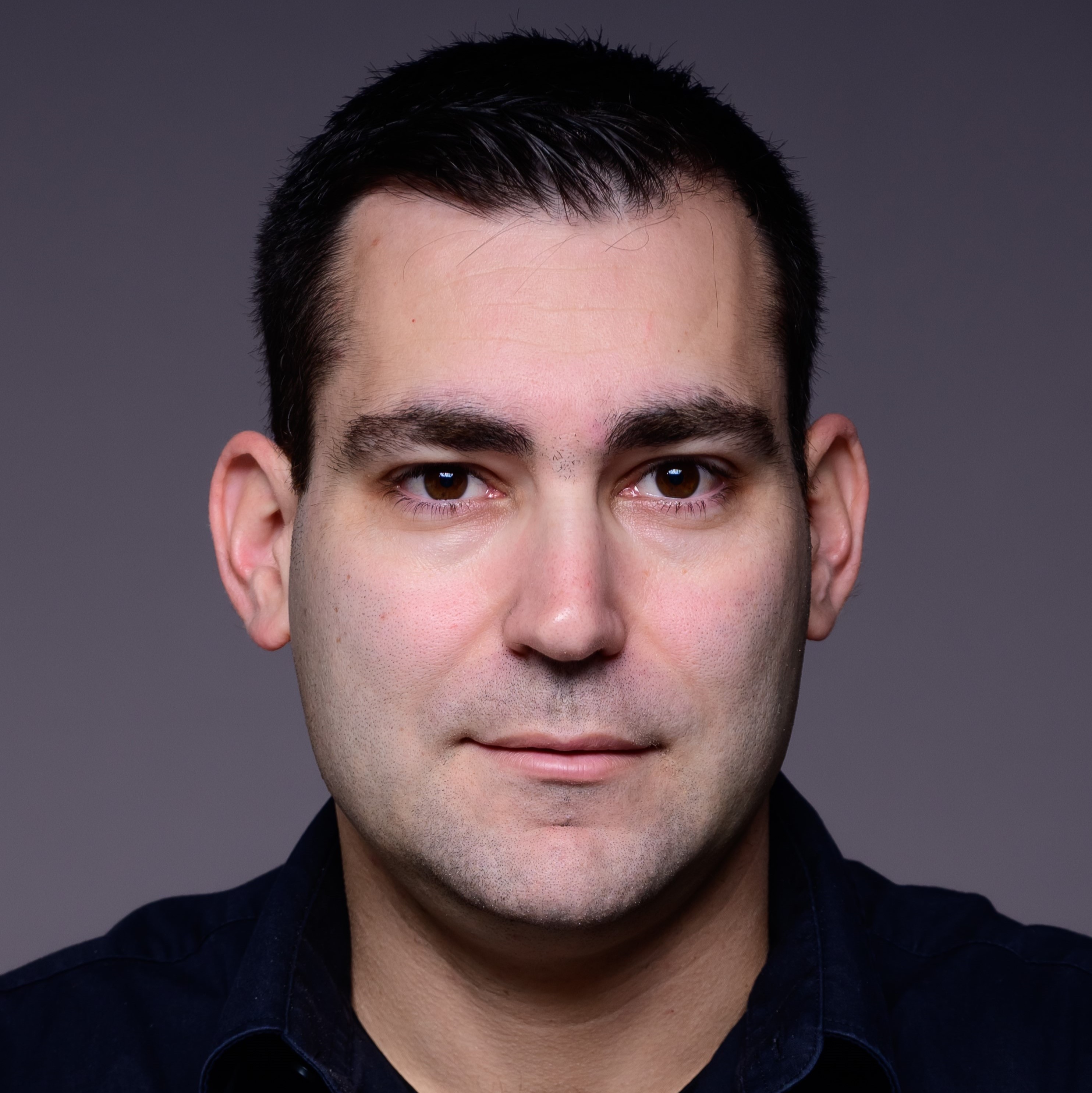}}]{Barak Or} (Member, IEEE) received a B.Sc. degree in aerospace engineering (2016), a B.A. degree (cum laude) in economics and management (2016), and an M.Sc. degree in aerospace engineering (2018) from the Technion–Israel Institute of Technology. He graduated with a Ph.D. degree from the University of Haifa, Haifa (2022).
His research interests include navigation, deep learning, sensor fusion, and estimation theory.
\end{IEEEbiography}

\end{document}